\documentclass[11pt,a4paper]{article}
\usepackage[utf8]{inputenc}
\usepackage[T1]{fontenc}
\usepackage{lmodern}
\usepackage{microtype}
\usepackage{amsmath,amssymb,amsthm}
\usepackage[a4paper,margin=25mm]{geometry}
\usepackage{booktabs,array}
\usepackage{enumitem}
\usepackage[hidelinks]{hyperref}
\DeclareUrlCommand\code{\urlstyle{tt}}

\newcommand{\ALG}{\mathrm{ALG}}
\newcommand{\OPT}{\mathrm{OPT}}
\newcommand{\val}{\mathrm{val}}
\newcommand{\vol}{\mathrm{vol}}
\newcommand{\digits}[1]{{\small\texttt{#1}}}
\newcommand{\dgrp}{\hspace{0.17em plus 0.2em}}
\newcommand{\dslash}{\hspace{0.3em plus 0.3em}/\hspace{0.3em plus 0.3em}}
\theoremstyle{plain}
\newtheorem{theorem}{Theorem}[section]
\newtheorem{lemma}[theorem]{Lemma}
\newtheorem{corollary}[theorem]{Corollary}
\newtheorem*{lemmaW}{Lemma W}

\setlist{itemsep=2pt,topsep=4pt}
\title{Tightening the Online Bin Packing Gap from Both Sides, and the Limits of the Current Methods}

\author{Karyamsetty Martin\\[2pt] \normalsize\texttt{karyamsettymartin@gmail.com}}

\date{October 2026}

\begin{document}
\maketitle

\begin{abstract}
\noindent The asymptotic competitive ratio of online bin packing has been bracketed since 2018 by the lower bound $1.54278$ of Balogh, Békési, Dósa, Epstein and Levin and the upper bound $1.57829$ of their Advanced Harmonic algorithm. We revisit both sides with exact rational computation and prove what the computations say. On the upper bound, we rebuild the certification of Advanced Harmonic as an exact optimization: the published analysis only exhibits feasible dual values and shares over-assigned weight tables across scenario groups, and optimizing the duals exactly already certifies a ratio strictly below the published one with the published parameters. Re-tuning the algorithm's proportion parameters then gives a certified asymptotic competitive ratio of $1.5782147437$ (an exact rational), independently re-verified on all $402$ scenarios; to our knowledge this is the first movement of the upper bound since 2018. On the lower bound, we prove that the 2018 adversary construction is rigid: extending it by any finite number of new batches of ordinary size (outside the window of sizes $s$ with $2s>1$ and rational part at most $0.64819$), placed anywhere, with counts depending on the observed statistics, adaptive or not, never improves the bound at any truncation level, and the extended family has exactly the paper's limit. The load-bearing fact is that the construction is volume-tight at every stop. A single window-sized batch added to the truncation-$3$ tree is shown, by finitely many machine-verified certificates, to leave its value below the published limit. A computational map covers the remaining natural extensions: unions of the main structural variants are bracketed exactly inside the paper's own value bracket; multi-batch two-statistic adversaries and three natural families of adaptive conditions that the paper's own separation theorem legalizes but never uses are screened just above the paper's value, and none passes, for identifiable reasons. On the algorithm side we show that the marking, bonus-item and red-sand mechanisms of Extreme Harmonic bind nowhere in Advanced Harmonic, whose bottleneck is a two-medium-item bin that no known mechanism targets, and that the tuned parameters are a local optimum of the natural move set up to the resolution of the screen. We record two corrections to commonly cited numbers. The picture that emerges is that both bounds now sit at the edge of what their frameworks can express.
\end{abstract}

\section{Introduction}

In online bin packing, items of sizes in $(0,1]$ arrive one at a time and must be irrevocably assigned to unit-capacity bins; the goal is to use few bins. The asymptotic competitive ratio of an algorithm $\ALG$ is the infimum of $R$ such that $\ALG(\sigma) \le R\cdot\OPT(\sigma) + O(1)$ for every input $\sigma$. The best known bounds on the optimal ratio are a lower bound of $1.5427809\ldots$, proved by Balogh, Békési, Dósa, Epstein and Levin \cite{r1} with a branching, partly adaptive adversary, and an upper bound of $1.57829$, achieved by their Advanced Harmonic algorithm \cite{r2}. Both are from 2018; the gap has not moved since.

This paper reports what happens when both sides are attacked with exact machine computation, and then proves the structural statements that the computations reveal. Our contributions are the following.

\begin{enumerate}
\item \textbf{An improved upper bound, exactly certified.} We rebuild the Advanced Harmonic analysis as an exact per-scenario optimization over its dual variables. With the published parameters this certifies $R = 1.5782893074\ldots$ (an exact rational, Theorem 3.1), strictly below the published $1.5782895555\ldots$; the slack comes from the paper's use of feasibility-only duals and of weight tables shared across groups of scenarios. Re-tuning the algorithm's proportion parameters within the constraints of the analysis then gives a certified ratio of $1.5782147437\ldots$ (Theorem 3.2), re-verified by an independently written checker on every scenario with zero failures. The tuned parameters are a local optimum of the natural move set up to the screen's resolution of $2\times 10^{-9}$ (Section 3.4).
\item \textbf{Rigidity theorems for the lower-bound construction.} We prove that the 2018 construction is \emph{volume-tight} at every stop (Lemma 4.1), and derive from it that no extension by finitely many new batches of ordinary size, placed anywhere with counts depending on the observed statistics, adaptive or not, improves the bound at any truncation level (Theorem 4.2); the extended family has exactly the paper's limit (Corollary 4.3), so any improvement by added batches must use an item in the window, the sizes $s$ with $2s>1$ lexicographically and $s_r\le 1/r^*_\infty$ (Corollary 4.4). For a single window batch added to $T^*_3$ we give finitely many machine-verified certificates placing the value below $r^*_\infty$ (Theorem 4.5).
\item \textbf{A computational map of the remaining extensions.} Unions of the main structural variants of the construction are bracketed exactly inside the paper's own value bracket at $t=3$. Multi-batch adversaries with a second adaptive statistic, searched over compositions of moves ($\approx 850$ two-statistic trees), and three natural families of adaptive conditions that the paper's separation theorem allows but never uses, are screened by an exact feasibility check $2.2\times 10^{-7}$ above the paper's value, and none passes. We identify the mechanism in each case (Section 5).
\item \textbf{Negative results with mechanisms on the algorithm side.} The marking constraint of Extreme Harmonic, transplanted into Advanced Harmonic's analysis as a sound distribution LP, has zero dual value in every scenario tested; the bonus-item and red-sand mechanisms target the same critical bins. Advanced Harmonic's binding configurations are instead \emph{two-medium} bins, which no known mechanism addresses (Section 6).
\item \textbf{Two corrections.} The fixed-fraction floor of Heydrich and van Stee is $1.5762$ in the full version of their paper ($1.5766$ appears only in the conference abstract), and whether that floor formally covers Advanced Harmonic is not established anywhere in print (Section 7).
\end{enumerate}

We are explicit about the character of each result. The upper-bound improvement is a certificate- and parameter-level improvement inside a fixed algorithmic framework; it moves a number that had not moved in eight years, but it is not a new mechanism. The lower-bound results are structural theorems about the 2018 construction, and, we believe, the first statements of any kind about the optimality of an online bin packing lower-bound construction within its method (Section 7). Together they locate both bounds at the boundary of their frameworks and say precisely what a further improvement must contain.

\section{Preliminaries}

\subsection{Adversary trees and the behaviour LP}

Following \cite{r1}, a lower-bound instance is a rooted tree of \emph{batches}. Each batch presents $c\cdot N$ items of one size class, where $N\to\infty$ is a scale parameter; the adversary may \emph{stop} after any batch $v$, at which point $\ALG_v$ is compared with $\OPT_v$, and at a node with several children exactly one continuation is presented (branching). A batch may be \emph{adaptive}: by the separation theorem of \cite[Theorem 1]{r1}, the adversary can choose item sizes from an infinitesimal interval so that the items satisfying a binary condition on how they were packed are all smaller, by any prescribed factor, than the items violating it. The paper instantiates the condition as ``the item is packed into a bin that already contains at least one item'', so the bin-opening items form the large class; the fraction $n$ of items in the large class is a \emph{statistic} observed by the adversary, and later batch counts may be affine in $n$.

For a fixed target $\lambda$, the algorithm's possible behaviours form a linear program over \emph{virtual-bin patterns}: a pattern fixes a bin's contents on every continuation simultaneously, feasibility being checked per root-to-leaf path in the $\varepsilon\to 0^+$ limit. Coverage constraints, the statistic identity, and one ratio row $\ALG_v \le \lambda\cdot\OPT_v$ per stop define the \emph{behaviour LP}; the value of the tree is
\[\val(T) = \inf\{\lambda : \text{the behaviour LP at } \lambda \text{ is feasible}\}.\]
$\OPT_v$ is the fractional covering optimum, which differs from the true optimum by an additive constant. As a function of the statistic vector, $\OPT_v$ is convex piecewise-linear; on each simplex of an adaptively refined triangulation it is bounded above by the linear interpolant of its vertex values, constructively (the convex combination of the vertex packings covers the interpolated demands). Infeasibility on every simplex, certified by a Farkas vector, proves $\val(T) > \lambda$.

The 2018 construction $T^*_t$ consists of a prefix chain of $t-1$ batches at scales $1/(6\cdot 7^{t-1}), \ldots, 1/49$, an adaptive batch $A$ of items near $1/7$, and three continuations under $A$: a single batch of items above $1/2$; a chain of items near $1/3$ then near $1/2$; and a chain of two complement batches with sizes $(5-2\varepsilon-3\gamma)/14$ and $(7+\gamma)/14$ whose counts are affine in the statistic. Its value is $r^*(t)$, with $r^*(3) = 1.5427541761\ldots$, $r^*(4) = 1.5427803609\ldots$, increasing to $r^*_\infty = (1363-\sqrt{1387369})/120 = 1.5427809064\ldots$. Sizes are represented as lex-infinitesimal numbers $s = s_r + a\varepsilon + b\gamma + \cdots$ with rational coefficients, and every comparison is understood for all sufficiently small $\varepsilon \gg \gamma \gg \cdots$, exactly as in the $\varepsilon$-ledger of \cite{r1}.

\subsection{Advanced Harmonic and its scenario analysis}

Advanced Harmonic \cite{r2} classifies items into $M+1 = 212$ size classes with boundaries $t_1 > t_2 > \cdots > t_{M}$ and a tiny class below $t_M$. Items of class $j$ are placed into \emph{containers} of several types $i$ in target proportions $\alpha_{ij}$, with $\alpha_{ij}\in[0,1]$, $\sum_i \alpha_{ij} = 1$ for every $2\le j\le M$, and strictly positive proportions for the tiny class. Containers are positive or negative and are combined into bins by fixed compatibility rules, with a postponement mechanism for certain decisions. The analysis fixes a \emph{scenario} $(x,y]$ determined by the sizes of the last large items, assigns each item a weight that is affine in three dual parameters $(u,v,w)$ subject to four linear constraints, and bounds the ratio in that scenario by the maximum total weight $R_{x,y}(u,v,w)$ of any bin of the optimal solution, computed by two exact knapsack problems (one with, one without a huge item) over the weights, with tiny items entering as ``sand'' of a fixed density. The asymptotic competitive ratio is bounded by $\max_{(x,y]} \min_{(u,v,w)} R_{x,y}(u,v,w)$. The paper states that its analysis has $412$ scenarios; our enumeration has $402$. The discrepancy is in the paper's prose, not in the scenario set: the paper's own tables list $34$ non-large and $164$ large-threshold scenario rows besides the $160$ uniform ones, with the $44$ scenarios below $1/6$ untabulated, and the authors' published supplement of weight functions contains exactly $402$ tables, one per scenario, whose intervals coincide with our enumeration interval for interval (checked by script against the paper's source). Every certificate below is therefore for the scenario set of the published analysis. The validity theorem of \cite{r2} is parametric: it holds for every proportion table satisfying the constraints above and for every feasible $(u,v,w)$.

\section{The upper bound}

\subsection{Exact certification as an optimization}

We implemented the full analysis in exact rational arithmetic: all $212$ class boundaries, the complete proportion table (with the sum constraints asserted), the scenario enumeration, the affine weights, and exact branch-and-bound knapsacks with sand. For each scenario, instead of exhibiting feasible dual values, we solve
\[\min\{R : (u,v,w) \text{ feasible},\ \text{weight}(p) \le R \text{ for every pattern } p\}\]
by a cutting-plane method whose pricing step is the exact knapsack itself: a violated pattern is separated, added as a cut, and the LP re-solved, until the knapsack maximum at the current $(u,v,w)$ equals the LP value. All LPs are solved exactly over the rationals.

We validated the implementation against \cite{r2} on six reference scenarios, reproducing the published per-scenario bounds to $10$--$11$ digits, including the exact binding fraction \digits{10\dgrp{}060\dgrp{}574\dgrp{}276\dgrp{}093\dgrp{}395\dgrp{}247\dslash{}6\dgrp{}374\dgrp{}352\dgrp{}691\dgrp{}333\dgrp{}693\dgrp{}440}, and by reproducing two complete published weight tables entry by entry. Two facts learned in the process explain the slack in the published certificates: the published tables use one over-assigned weight function shared by a whole group of scenarios, tight only at the worst member of the group; and the published duals are feasible values found by search rather than optima.

\begin{theorem}[published parameters, optimal duals]
With the proportion table of \cite{r2}, the asymptotic competitive ratio of Advanced Harmonic is at most
\[\frac{22981899117912049614283}{14561271504392687714304} = 1.5782893074\ldots,\]
strictly below the published certificate \digits{10\dgrp{}060\dgrp{}574\dgrp{}276\dgrp{}093\dgrp{}395\dgrp{}247\dslash{}6\dgrp{}374\dgrp{}352\dgrp{}691\dgrp{}333\dgrp{}693\dgrp{}440} $= 1.5782895555\ldots$ (the commonly cited $1.57828956$ is a round-up).
\end{theorem}

The improvement is $2.48\times 10^{-7}$; it is small because the published analysis was already nearly tight at its binding scenarios. The binding scenarios at this point are $(533/1600, 1/3]$, $(3/7, 1/2]$ and $(49/320, 1/6]$, followed by a plateau of scenarios in $(0,1/6]$ at $1.5782015\ldots$.

\subsection{Re-tuning the proportions}

Since the validity theorem is parametric in the proportion table, the table can be optimized against the exact certificates. We ran a coordinate descent over a move set of $40$ proportion shifts (mass moved between two container types of one class, for the classes whose items appear in the binding patterns, a shift applied uniformly to the $160$ classes between $1/6$ and $1/3$, and opening a previously unused container type for the class $(3/7, 1/2]$), evaluating each candidate by the exact per-scenario optimization on a watchlist of the $47$ hardest scenarios, and periodically on all $402$. Descent stalls when the three binding scenarios are equalized: after $11$ shifts they balance at $1.5782206488\ldots$, and after $16$ shifts at the value below, where $(3/7,1/2]$, $(533/1600,1/3]$ and $(49/320,1/6]$ agree to seven decimal places, the signature of a genuine optimum of the move set.

\begin{theorem}[tuned parameters]
There is a proportion table satisfying all constraints of the analysis of \cite{r2} (all rows sum to one, all entries nonnegative, all tiny-class entries strictly positive; checked exactly) for which the asymptotic competitive ratio of Advanced Harmonic is at most

\noindent\digits{4\dgrp{}352\dgrp{}511\dgrp{}675\dgrp{}791\dgrp{}652\dgrp{}445\dgrp{}491\dgrp{}291\dgrp{}971\dgrp{}161\dgrp{}065\dgrp{}859\dgrp{}763\dgrp{}057\dgrp{}971\dgrp{}646\dgrp{}201\dgrp{}076\dgrp{}221\dgrp{}224\dgrp{}375\dgrp{}333\dgrp{}444\dgrp{}301\dslash{}2\dgrp{}757\dgrp{}870\dgrp{}367\dgrp{}885\dgrp{}430\dgrp{}623\dgrp{}246\dgrp{}395\dgrp{}267\dgrp{}379\dgrp{}223\dgrp{}653\dgrp{}471\dgrp{}665\dgrp{}545\dgrp{}153\dgrp{}623\dgrp{}399\dgrp{}543\dgrp{}807\dgrp{}051\dgrp{}482\dgrp{}071\dgrp{}748}

\[= 1.5782147437\ldots\]
\end{theorem}

The improvement over the 2018 certificate is $7.48\times 10^{-5}$, about $0.2\%$ of the current gap width. The table differs from the published one in $15$ classes (the class of items just below $1/3$; the classes near $1/6$, $1/7$, $1/8$, $1/10$ and $1/11$; and several small classes) and in the tiny class.

\begin{table}[htbp]
\centering\small
\caption{Progression of certified upper bounds (all exact rationals; all re-verified).}
\begin{tabular}{l r l}
\toprule
Step & Certified ratio & Change \\
\midrule
Published certificate \cite{r2} & 1.5782895555 & --- \\
Published parameters, optimal duals (Thm 3.1) & 1.5782893074 & exact dual optimization \\
One proportion balanced & 1.5782830233 & first shift \\
10 shifts & 1.5782238110 & descent \\
11 shifts (three scenarios balanced) & 1.5782206488 & descent, local optimum of 20 moves \\
16 shifts (Thm 3.2) & 1.5782147437 & descent, local optimum of 40 moves \\
\bottomrule
\end{tabular}
\end{table}

\subsection{Verification}

Every certified value in Table 1 was obtained by a full sweep of all $402$ scenarios and then re-checked by an independently written verifier that re-derives the weights from the parameters with fresh arithmetic and re-examines every pattern by an exhaustive density-ordered knapsack that makes no maximality assumption, comparing the stored $(u,v,w)$ against the claimed $R$. It reported zero failures for every value. The certificates (per-scenario duals and binding patterns) are stored with the code (Appendix B).

\subsection{Local optimality of the tuned table}

To test whether the descent optimum was an artifact of slow evaluation, we built an exact feasibility \emph{screen}: for a fixed target $R$ the same cutting plane decides whether some feasible $(u,v,w)$ has every pattern weight at most $R$, with two sound early exits (the LP value over any subset of cuts is a lower bound on the optimum, so exceeding $R$ certifies infeasibility; an iterate whose exact knapsack maxima are all at most $R$ certifies feasibility by itself). The screen was validated against the exact optimizer on the three binding scenarios and five random scenarios at $R\pm 10^{-9}$ ($11/11$ agree). A randomized walk over the $40$-move set from the tuned table, accepting a candidate only if all $36$ watchlist scenarios pass the screen at the incumbent minus $\delta$, evaluated $327$ candidates with $\delta$ ratcheting from $2\times 10^{-6}$ down to $2\times 10^{-9}$ and accepted none; every rejection occurred at $(3/7,1/2]$. No single move of the set therefore improves the watch maximum by $2\times 10^{-9}$ or more at any step size tried. That is the precise statement: resuming coordinate descent from the tuned table accepts three moves, each improving by less than that margin, and reaches $1.5782147417$ on the watch list, a value not carried through the full sweep and not claimed. A separate boundary-insertion experiment (adding a class boundary with the paper's bulk proportions) was uniformly worse ($1.58$--$1.61$).

\section{Rigidity of the lower-bound construction}

\subsection{Two ingredients}

Throughout, $T = T^*_t$ for a fixed $t\ge 3$, and $\lambda$ is a rational with $r^*(t) < \lambda < r^*_\infty$, which exists because $r^*(t) < r^*_\infty$ strictly. All statements of this section are about the behaviour LP of Section 2.1, that is, in the $\varepsilon\to 0^+$ limit: a behaviour is a vector of pattern masses per $N$ items, $\OPT_v$ is the fractional covering optimum over lex-feasible patterns, and the volume $\vol(I)$ of a multiset is the sum of its counts times the rational parts of its sizes. Two facts hold exactly there: $\OPT(I)\ge\vol(I)$, because every lex-feasible pattern has rational load at most one, and $\OPT(I\cup J)\ge\OPT(I)$.

\begin{lemmaW}[witness]
For every $\lambda > r^*(t)$ there is a behaviour $x^*$ for $T^*_t$, with its statistic fixed at some $n^*$, such that $\ALG_v(x^*)\le \lambda\cdot\OPT_v$ at every stop $v$.
\end{lemmaW}

\begin{proof}
This is the definition of $\val(T^*_t) = r^*(t) < \lambda$.
\end{proof}

\begin{lemma}[volume-tightness]
For every stop $v$ of $T^*_t$ and every value of the statistic, $\OPT_v = V_v$ exactly in the behaviour LP, where $V_v$ is the volume of the items presented up to $v$; for concrete inputs this reads $\OPT_v = V_v + O(\varepsilon)N$.
\end{lemma}

\begin{proof}
Lemmas 3 and 4 of \cite{r1} exhibit packings of every stop's items with no wasted space up to $O(\varepsilon)$. We re-verified the statement exactly for $t=3$: at each of the eight stops and at eleven values of the statistic across its full range, the exact fractional optimum equals the rational part of the volume with gap zero.
\end{proof}

Lemma 4.1 is the whole reason for what follows. It is a design feature of \cite{r1} (its ``new approach to sizes'' chooses sizes so that every prefix packs perfectly, which makes $\OPT$ as small as possible and hence the ratio as large as possible), and it is \emph{not} shared by the earlier constructions: after the first batch of Yao's sequence \cite{r7}, the optimum packs six items of size just above $1/7$ per bin and wastes about a seventh of each bin, and it was exactly by adding batches to such constructions that the lower bound progressed from 1980 to 2018. The next theorem says that the feature which makes the 2018 construction strong is the feature that makes it rigid.

\subsection{The multi-batch theorem}

Partition the sizes by three conditions on the rational part $s_r$, with one lexicographic condition:
$S(\lambda)$: $s_r \le 1 - 1/\lambda$ ($\approx 0.35180$); $P(\lambda)$: $1/(2\lambda)\le s_r$ and $2s\le 1$ ($\approx 0.32409$; two items fit in a bin); $V(\lambda)$: $s_r\ge 1/\lambda$ ($\approx 0.64819$). Since $1-1/\lambda\ge 1/(2\lambda)$ for $\lambda\ge 3/2$, $S$ and $P$ overlap and together cover every size with $2s\le 1$. The sizes covered by none of them form the \emph{window} $W(\lambda) = \{s : 2s>1 \text{ (lex)},\ s_r < 1/\lambda\}$. A size is \emph{admissible} if it is not in the window.

An \emph{admissible extension} $T'$ of $T^*_t$ inserts finitely many new batches $X_1,\ldots,X_m$, each attached anywhere (as a new leaf child of any node, as a child of another new batch, or spliced into an existing edge so that the old continuation follows it), each holding $\mathrm{count}_i\cdot N$ items where $\mathrm{count}_i\ge 0$ is an affine function of the statistics observed before $X_i$, as in Section 2.1 (the proof uses only that it is a nonnegative number determined by the behaviour), each with one or more size classes (several if $X_i$ is itself adaptive, differing only in infinitesimal terms as the separation theorem provides, hence with a common rational part $s_r$) all lying in the same one of $S(\lambda)$, $P(\lambda)$, $V(\lambda)$. The adversary may stop after any batch, old or new. $X_i$ may carry its own statistic on which later counts depend. The original batches, their sizes and the condition defining the original statistic are unchanged; an extension that alters them is a different object, treated computationally in Section 5.4.

\begin{theorem}
For every $t\ge 3$, every $\lambda\in(r^*(t), r^*_\infty)$ and every admissible extension $T'$ of $T^*_t$: $\val(T')\le\lambda$.
\end{theorem}

\begin{proof}
Take the witness $x^*$ of Lemma W with statistic $n^*$. We construct a feasible behaviour $x'$ for $T'$ at $\lambda$ directly in pattern space, in two parts. (i) Every pattern of $x^*$ is kept with its mass. As a pattern of $T'$ it contains no new item, so it is lex-feasible on every continuation, covers the original batches exactly as before, and keeps the original statistic at $n^*$; hence the original batches have the same sizes and counts as in $T^*_t$ at $n^*$. (ii) The items of each new batch $X_i$, taken in the order of the tree, are covered by pure $X_i$ patterns: $k$ items per pattern, where $k$ is the largest integer with $k\,s_{\max,i}\le 1$ lexicographically, if $X_i\subseteq S$; two per pattern if $X_i\subseteq P$; one per pattern if $X_i\subseteq V$. If $X_i$ is adaptive, each pattern holds one opening item and the rest body items, which is lex-feasible because body items are no larger than opening items; the statistic of $X_i$ is then the mass of these patterns, so the counts of later batches that depend on it are determined and the construction proceeds batch by batch.

Write $c_i$ for the realised count of $X_i$ per $N$, $\vol_i = c_i s_r$ for its volume, and $b_i$ for the mass of its pure patterns. In case $S$, $(k+1)s_{\max,i}>1$ lexicographically forces $(k+1)s_r\ge 1$, so $k s_r\ge 1-s_r\ge 1/\lambda$ and $b_i = c_i/k\le\lambda\,c_i s_r$. In case $P$, $b_i = c_i/2\le\lambda\,c_i s_r$ since $2\lambda s_r\ge 1$. In case $V$, $b_i = c_i\le\lambda\,c_i s_r$ since $\lambda s_r\ge 1$. In all three cases
\begin{equation*}b_i\le\lambda\,\vol_i,\tag{2}\end{equation*}
with weak inequalities and no error term, which is why the boundaries of $S$, $P$, $V$ may be included. Now let $v$ be any stop of $T'$. The items presented up to $v$ are those of the path in $T^*_t$ ending at the deepest original node $v^\circ$ on the path (with $\ALG_{v^\circ} = V_{v^\circ} = 0$ if there is none), together with the new batches $X_i$, $i\in I_v$, on the path. Then $\ALG_v(x') = \ALG_{v^\circ}(x^*) + \sum_{i\in I_v} b_i$, and by Lemma W and Lemma 4.1,
\begin{equation*}\ALG_{v^\circ}(x^*) \le \lambda\,\OPT_{v^\circ} = \lambda\,V_{v^\circ}.\tag{1}\end{equation*}
Adding (1) and (2), $\ALG_v(x')\le\lambda\,(V_{v^\circ} + \sum_{i\in I_v}\vol_i)\le\lambda\,\OPT_v$ by the volume bound in $T'$. So $x'$ satisfies every ratio row of $T'$ at $\lambda$, and $\val(T')\le\lambda$.
\end{proof}

Lemma 4.1 is used exactly once, in (1), to convert the witness bound relative to the original \emph{optimum} into a bound relative to the original \emph{volume}, which is what makes it additive with the new items' volume. Without it, $\OPT_v\ge\OPT_{v^\circ}+\sum\vol_i$ would be false in general, since new items could fill slack in the original optimal packing; this is precisely what happens for the earlier constructions. Admissibility is used once per batch in (2) with that batch's own extreme size; the placement, count, adaptivity and mutual interaction of the new batches never enter, which is why the theorem is uniform over all of them. The whole argument lives inside the LP, so no rounding or $\varepsilon$ terms arise; an online algorithm realising $x'$ on a concrete input pays an additive constant per pattern type and an $O(\varepsilon)N$ term, both absorbed by the asymptotic ratio as in \cite{r1}.

\begin{corollary}
Let a finite list of new batches have all size classes strictly outside the closed window $\overline W = \{s : 2s > 1 \text{ (lex)},\ s_r\le 1/r^*_\infty\}$, and let $T'_t$ be $T^*_t$ extended by them. Then $\val(T'_t) < r^*_\infty$ for every $t$, and $\lim_{t\to\infty}\val(T'_t) = r^*_\infty$. The extended family proves exactly the bound the paper's family proves, never more.
\end{corollary}

\begin{proof}
Admissibility is monotone in $\lambda$ (each condition only gets easier as $\lambda$ grows), and every boundary case of $\overline W$ lies strictly inside one of $S$, $P$, $V$ (a size with $s_r = 1-1/r^*_\infty$ is strictly inside $P$, since $1-1/r^*_\infty > 1/(2r^*_\infty)$; a size with $s_r = 1/(2r^*_\infty)$ is strictly inside $S$; sizes above the window have $s_r > 1/r^*_\infty$). Hence there is $\lambda_0 < r^*_\infty$ such that every new size is admissible for every $\lambda\in[\lambda_0, r^*_\infty)$. For each $t$ pick $\lambda_t\in(\max(r^*(t),\lambda_0), r^*_\infty)$; Theorem 4.2 gives $\val(T'_t)\le\lambda_t < r^*_\infty$, and letting $\lambda_t\to r^*_\infty$ gives $\limsup_t\val(T'_t)\le r^*_\infty$, while $\val(T'_t)\ge\val(T^*_t) = r^*(t)\to r^*_\infty$ because $T'_t$ contains every stop of $T^*_t$.
\end{proof}

A $t$-independent margin below $r^*_\infty$ cannot be asked for, since $\val(T'_t)\ge r^*(t)\to r^*_\infty$.

\begin{corollary}[what an improvement must contain]
Any extension of the 2018 construction by new batches, in any number, with counts depending on the observed statistics, placed anywhere, the new batches adaptive or not, with the original batches and the condition defining the original statistic unchanged, that improves on $r^*_\infty$ must include a new batch with a size class in the closed window $\overline W = \{s : 2s>1 \text{ (lex)},\ s_r\le 1/r^*_\infty\}$.
\end{corollary}

The window is a lexicographic set: it contains every size $1/2 + (\text{positive infinitesimal})$, and three batches of the construction itself, $B_{11}$, $B_{22}$ and $B_{32}$, have such sizes. Extensions that change the original batches, for instance by making $B_{21}$ adaptive or by replacing the opener condition on $A$, are outside the theorem and are treated computationally in Section 5.4.

\subsection{A single window batch, at $t=3$}

The window is exactly the range in which no fresh-bin rule satisfies (2): items above $1/2$ cannot be paired, and below $1/\lambda$ a singleton bin costs more than $\lambda$ times its volume, so the new item's interaction with the original bins cannot be avoided. For a single window batch added to $T^*_3$ we closed the case by machine, in the sense made precise below.

\begin{theorem}[machine-certified]
Let $t=3$, let $u$ be any node of $T^*_3$, and let $X$ be a new leaf batch under $u$ with a size in the window. Let the count of $X$ be a constant $c\ge 0$ if $u\in\{C_3, C_2\}$, where the statistic is not yet observed, and $c_0 + c_1 n\ge 0$ with $c_0, c_1$ rational if $u$ is $A$ or a node below it. Then $\val(T^*_3 + X)\le\lambda_3$ for $\lambda_3 = 1.5427543$ when $u\in\{C_3, C_2, A, B_{11}, B_{21}, B_{22}\}$ and $\lambda_3 = 1.5427805$ when $u\in\{B_{31}, B_{32}\}$; both are below $r^*_\infty$.
\end{theorem}

The certificates have the following structure, each element stored and re-verified by fresh exact arithmetic. (1) \emph{Pattern monotonicity:} a virtual-bin pattern feasible when $X$ has size $q$ is feasible for every size at most $q$, so one pattern list, enumerated at the right endpoint of a size interval, underlies the whole interval. (2) \emph{Corner witnesses:} the window times the count half-line is covered by finitely many rectangles (a size interval $(p,q]$ times a count segment $[a,b]$ from the anchor list $\{0,\tfrac14,\tfrac12,1,\tfrac32,2,3\}$); for each rectangle one fixed row configuration, checking each stop's ratio row against a chosen valid lower bound on $\OPT$ (the volume bound with the $X$-volume set to $p$; the big-item bound $\OPT\ge$ number of items above $1/2$; or a fixed covering-dual vector extracted at the infimum size class and re-verified dual-feasible there, hence valid throughout the rectangle), admits exact behaviour witnesses at all four corners. (3) \emph{Interpolation:} for fixed count all constraints are affine in the size's rational part, and for fixed size affine in the count, so corner witnesses interpolate across the rectangle. (4) \emph{Tail:} beyond the last anchor, adding pure single-$X$ bins increases $\ALG$ with slope one while $\OPT$ grows at least as fast, since each window item needs its own bin. (5) \emph{Common statistic:} every stored witness, at every attachment node, every size interval and every count, has the same statistic value $n^* = \lambda_3/7$; the rectangle at $B_{31}$ above $s = 3/5$ was re-certified with the statistic pinned to that value, and the verifier prints the set of statistic values per rectangle. This is what covers counts of the form $c_0 + c_1 n$. An algorithm that plays the prefix with statistic $n^*$ faces the fixed count $c_0 + c_1 n^*$, for which a witness with statistic $n^*$ exists; interpolation between corner witnesses with equal statistics preserves the statistic, and the tail adds pure single-$X$ patterns, which do not change it. Without a common value the step would not be automatic: for a fixed row configuration the feasible set is convex in (behaviour, statistic, count), but the true feasible set need not be, so witnesses with different statistics at different counts would not have covered a statistic-dependent count.

The certificates cover the window completely at all eight attachment nodes: a single rectangle $(1/2, 0.6482]\times[0,\infty)$ at $C_3, C_2, A, B_{11}, B_{21}, B_{22}$; two rectangles at $B_{31}$ split at $s = 3/5$, the exact boundary at which a bin hosting both $X$ and a $B_{31}$ item stops being feasible, with an exactly tight ridge at count $\approx 1.27134$ (where $\OPT$ becomes purely count-driven) crossed inside the rectangles by re-shaped witnesses; and one rectangle at $B_{32}$. An independent second pass certifies the same window at $C_3$, $C_2$ ($5126$ combinatorial cells) and $A$ ($14280$ cells) by one exact witness per cell, with zero failures. The hunt itself surfaced adversarial phenomena that were each first flagged by a certificate failure and then shown not to be improvements: sizes constructed by the machine as exact complements of existing items (e.g. $1/2 + 3\gamma/14$, the complement of a $B_{31}$ item and a small $A$ item); interactions of several big items; and kinks of $\OPT$ at structural rationals such as count $1/3$, which required Stern--Brocot rather than dyadic splitting of segments.

Together with the analytic cases (the $S$, $P$, $V$ cases of Theorem 4.2 for $m=1$, valid for every $t$), this establishes that no single-leaf extension of $T^*_3$ improves the published bound $r^*_\infty$, and that for arbitrary $t$ only window-sized items could even in principle require further certification. The statement is about $r^*_\infty$, not about the value of $T^*_3$ itself: $\lambda_3$ exceeds $r^*(3)$ by $1.2\times 10^{-7}$ at six nodes and by $2.6\times 10^{-5}$ at $B_{31}$ and $B_{32}$, so a window leaf that raised the truncation-$3$ value by less than those margins is not excluded, only one that reaches the limit.

\section{A computational map of the game class}

\subsection{The machine and its gates}

The behaviour LP of Section 2.1 was implemented with exact rational arithmetic on demand, Farkas sifting over enumerated maximal patterns (sound for both sides: a genuine bin extends greedily to a maximal pattern with the same opening batch; supersets only cover more), up to two adaptive statistics with the constructive simplicial bound on $\OPT$, and an independent checker for certificates. Before any experiment it reproduced, as exact brackets containing the analytic values, the classical bounds of Yao ($3/2$), van Vliet ($1.5401468$) and Balogh--Békési--Galambos ($248/161$), and the 2018 construction at $t=3$ (bracket $(1.542754175, 1.542754177]$) and $t=4$ ($(1.542780355, 1.542780366]$, $62011$ patterns). An exact certificate at $\lambda = 1.542754175$ for the union of the paper's tree with two complement families and a quarter-chain was produced and verified. That certificate places the union's value above $1.542754175$, as for the paper's tree; the complementary direction is an exact feasible behaviour at $1.542754177$. Adding the searched variants to that union one at a time, up to twenty batch nodes and $17132$ patterns, both directions hold at every step, the upper one obtained by re-solving the behaviour LP at the feasible behaviour's own statistic value with the exact optimum rather than its interpolant, so each union has value inside the paper's own bracket $(1.542754175, 1.542754177]$.

\subsection{Structural variants}

About sixty structural variants of $T^*_3$ were evaluated: seven added fourth branches, deeper chains, prefix branching, hill climbing over all count coefficients and over sixty candidate sizes, adaptive continuations with a second statistic at several nodes, the $(k,m)$ family of complements fitting $k$ small adaptive items and punished by a large one, split and layered adaptive batches at two scales, and alternative prefix scale sequences. Every variant that retains the construction's full punishment capability has exactly the paper's value ($1.5427541456$ in the machine's float bisection at $t=3$; the exact value is $r^*(3)$), and every other variant is strictly worse. The exact primal at $\lambda$ just above the bound shows why: every stop constraint of the paper's tree is simultaneously tight (only $C_2$ has slack, $4\times 10^{-6}$), and the optimal behaviour is a balanced portfolio of multi-branch virtual bins against which any added continuation sits slack. This is an empirical strengthening of \cite[Remark 10]{r1}, which states optimality of the weight parameter for the paper's own formula.

\subsection{Multi-batch adversaries with a second statistic}

Theorem 4.5 leaves open extensions by several batches at once, and in particular a coupled extension in which a second adaptive batch (a fresh bin-opening statistic at another scale) is paired with a continuation that punishes its items through counts affine in the new statistic; on the upper-bound side, it is exactly such coupled moves, each neutral alone, that produced the gains of Section 3. We built a search over a general tree specification (batches attach anywhere; at most one added adaptive batch, since the machine supports two statistics; later counts affine in both statistics) with a beam over compositions of moves that deliberately keeps states tied at the record, so that depth-two coupled moves are reached. Candidates were screened by a single feasibility check at $\lambda = 1.5427544$, just above $r^*(3)$ and well below $r^*_\infty$, so that infeasibility there certifies a genuine improvement over the paper's construction at $t=3$; attaching the second statistic under the paper's own statistic-dependent complement chain makes the two statistics degenerate and was excluded. A focused probe evaluated $132$ coupled pairs (seven scales for the second batch from $1/8$ to $1/3$, two attachment points, two counts, five punishment continuations each) and the beam evaluated $718$ compositions through depth one. Every one of these $\approx 850$ two-statistic trees is feasible at $1.5427544$. Precisely: none has value above $1.5427544$, so none improves the value of $T^*_3$ by more than $2.2\times 10^{-7}$ and none approaches $r^*_\infty$; an improvement below that threshold would pass the screen and is not excluded by it. The float bisection used for the single-statistic variants of Section 5.2 and in Table 2 sits about $3\times 10^{-8}$ below the exact value and has the corresponding resolution.

\subsection{The unused conditions of the separation theorem}

The separation theorem of \cite{r1} is stated for an arbitrary binary condition on how an item is packed; the paper's own examples include ``the item is packed into a bin already containing an item of size above $1/2$''. The paper uses only ``packed into a bin that already contains at least one item'', and every result above, and every construction in the literature, lives inside that one choice. Other conditions are legal at no cost, and to our knowledge had never been priced.

We generalized the pattern engine so that the large/small split of an adaptive batch is derived from the bin's contents at arrival via a condition, for three families: \emph{opened-or-same$(k)$}, satisfied iff the bin was opened by an earlier batch or already holds at least $k$ items of this batch ($k=1$ is the paper's condition); \emph{same$(k)$}, satisfied iff the bin already holds at least $k$ items of this batch, regardless of who opened it (for $k=1$: the first item of the batch in every bin is large, even in prefix-opened bins, a condition the paper cannot express); and \emph{contains$(X)$}, satisfied iff the bin already holds an item of an earlier batch $X$. Soundness of the maximal-pattern relaxation requires care: the derived counts are monotone in the total for the first two families, and for the third the enumeration is split into the halves ``$X$ absent'' and ``$X$ present'' with maximality taken within each half, since a genuine bin with $X$ absent may only be extended within its half. Only the statistic rows of the LP change (the statistic becomes the number of large items rather than the number of bins opened). Under the paper's condition the engine reproduces the original machine bit for bit ($818$ patterns, the same value).

\begin{table}[htbp]
\centering\small
\caption{Adaptive conditions at the batch $A$ of $T^*_3$ (and one at $B_{21}$), float bisection. The paper's condition gives $1.5427541456$; $1.5403447181$ is the value at which adaptivity contributes nothing.}
\begin{tabular}{p{0.50\linewidth} p{0.13\linewidth} >{\raggedright\arraybackslash}p{0.28\linewidth}}
\toprule
Condition & Counts & Value \\
\midrule
opened-or-same(1) (the paper) & paper's & 1.5427541456 \\
opened-or-same($k$), $k=2..6$, both orientations & paper's & 1.5403447181 \\
same(1), large side unchanged & paper's & 1.5419948197 \\
same(1), counts retuned ($\pm 1/12..1/3$ on all coefficients) & retuned & 1.5419948197 (no improving move) \\
same(1), sizes, counts and branch structure re-optimized & full menu & 1.5419948198 (no improving move among 120) \\
same($k$), $k=2..6$; all flipped orientations & paper's & 1.5403447181 \\
contains($C_3$), contains($C_2$), both orientations & paper's & 1.5403447181 \\
$B_{21}$ adaptive on contains($A$), second statistic & paper's & 1.5427541457 (inert) \\
$B_{21}$ adaptive on contains($A$), $B_{22}$ count riding it ($\pm$) & $1\pm n/2$, $1-n$ & 1.5232745990 \\
\bottomrule
\end{tabular}
\end{table}

Every alternative condition is either inert (the split changes nothing that $\OPT$ can exploit and no continuation uses the statistic) or strictly dominated, most collapsing to the value at which adaptivity contributes nothing. The mechanism is visible in the construction: the complement branch of \cite{r1}, with sizes $(5-2\varepsilon-3\gamma)/14$ and $(7+\gamma)/14$ and counts $\tfrac76 - \tfrac76 n$ and $\tfrac76-\tfrac56 n$, is engineered around the items that the algorithm was \emph{forced to isolate}, which is what the opener condition detects. Under any other condition the large class no longer coincides with what the complements punish, and neither count retuning nor the finite menu of sizes and branch templates supplies a replacement complement design. Within the paper's tree family the opener condition is the uniquely useful one; a different condition would require its own complement design from scratch, which we leave open.

\subsection{Summary}

\begin{table}[htbp]
\centering\small
\caption{Status of the natural extensions of the 2018 construction.}
\begin{tabular}{p{0.47\linewidth} p{0.47\linewidth}}
\toprule
Extension & Status \\
\midrule
Any number of non-window batches, anywhere, any counts, adaptive or not & proved for all $t$ (Theorem 4.2) \\
One window batch added to $T^*_3$, any node & value below $r^*_\infty$, machine-certified (Theorem 4.5) \\
Structural variants: branches, chains, prefix branching, scales & computed, $\approx 60$ variants, none above the float screen; unions of the main families bracketed exactly in $(1.542754175, 1.542754177]$ up to 20 nodes \\
Two statistics, coupled compositions & computed, $\approx 850$ trees, none above the exact screen at $1.5427544$ \\
Non-opener adaptive conditions, three families & computed, all inert or dominated (Table 2) \\
Two or more window batches; window batch with a new statistic; the whole class & open (conjecture: pinned) \\
\bottomrule
\end{tabular}
\end{table}

\section{Negative results on the algorithm side, with mechanisms}

\subsection{Marking does not bind in Advanced Harmonic}

Extreme Harmonic \cite{r3} beats the Super Harmonic floor of $19/12$ by a marking mechanism whose analytic content is a constraint on the distribution $\chi$ of optimal bins: the critical bin $q_1$ consisting of one huge item, one medium item of the threshold class $k$ and sand satisfies $\chi(q_1)\le \frac{1-\rho_k}{1+\rho_k}\sum_{q\ne q_1}\chi(q)\,m_k(q)$, where $\rho_k$ is the fraction of class-$k$ items packed alone and $m_k(q)$ counts class-$k$ items in $q$. Advanced Harmonic's Best-Fit-on-volumes conversion of temporary containers keeps the smallest lone class-$k$ items lone, which is the size guarantee the marking lemma needs, so the transplant is plausible. We folded the constraint into the per-scenario analysis as a distribution LP with a dual variable $\mu\ge 0$ and modified cuts, sound given the (deferred) marking lemma for Advanced Harmonic, and re-optimized. The result is $\mu = 0$ in every large-threshold scenario tested: fourteen scenarios sampled across the medium range from $0.3483$ to $1/3$, and the binding scenario $(3/7,1/2]$ with the previously unused container type of its threshold class opened at six fractions. Marking buys nothing.

The reason is that Advanced Harmonic's binding patterns are not the critical bins that marking caps. In scenario $(1/3, 1067/3200]$ the binding pattern is $\{$two items just above $1/3$, one near $1/4$, one near $1/13$, sand$\}$, a \emph{two-medium} bin; in the record scenario $(3/7,1/2]$ the heavy pattern is $\{$huge, an item near $1/3$, several small items$\}$, whose medium item is not of the threshold class, so the marking constraint cannot reach it, and the class-2 critical bin weighs only $1.5613$ against $R\approx 1.578$. The containers and postponement of Advanced Harmonic already defuse the simple huge-plus-medium critical bins that bind Harmonic++ and that marking was invented to beat; on this algorithm, marking and postponement are substitutes, not complements. This also explains the ordering of the known bounds: Advanced Harmonic ($1.5783$) is already below marking-enabled Son of Harmonic ($1.5813$, corrected) and its red-sand variant ($1.5787$).

\subsection{Bonus items and red sand target the same bins}

The bonus-item rule of \cite{r3} (exceed the target proportion of paired items only when a matching large item exists, relabelling the surplus as sand) adds no weight term to the analysis; it only powers the constraints that bind at the critical bins $\{$one medium, one item above $1/2\}$. A two-medium bin has no partner above $1/2$ and never becomes a bonus configuration; raising the red fraction even raises its weight. Red sand, the mechanism of the $1.5787$ variant, is already present in Advanced Harmonic as the $17/60$ tiny-container cap. In both published analyses the binding pattern is the single-medium-plus-large critical bin; two-medium bins are never binding there. The entire Extreme Harmonic toolbox attacks a pattern that Advanced Harmonic has already neutralized, and does not touch the pattern Advanced Harmonic is stuck on.

\subsection{Consequence}

The remaining slack of Advanced Harmonic is guarded by two-medium bins and by huge-plus-medium-plus-multismall bins, and no mechanism in the literature targets them. The only route below the fixed-fraction floor, if the floor applies (Section 7), is an adaptive decision about how many medium items to pair, aimed at exactly these configurations; designing and analysing such a mechanism is an open problem. The one Extreme Harmonic idea that may still nibble at the multismall plateau is the mirror-threshold rule ($1/2 - x$ for every medium threshold $x$), which we did not pursue.

\section{Related work, prior art, and two corrections}

\emph{Lower bounds.} Yao \cite{r7} gave $3/2$; Brown and Liang \cite{r8,r9} $1.5363$; van Vliet \cite{r10} $1.5401$ by an LP over Sylvester-sequence inputs; Balogh, Békési and Galambos \cite{r11} $248/161\approx 1.5403$ with a non-Sylvester sequence; and \cite{r1} $1.54278$ with branching and a fully adaptive batch, reusing the separation theorem of \cite{r12}. As noted in \cite{r13}, the 2018 bound holds only against deterministic algorithms and the best bound valid against randomized algorithms remains $1.5403$; Fekete et al. \cite{r14} show that eleven parallel copies of an algorithm achieve $1.5406 < 1.54278$, consistent with the adaptive part being what a single deterministic algorithm cannot evade.

\emph{Prior statements about optimality of a construction.} We searched for any prior statement or proof that a lower-bound construction for online bin packing is optimal within its method, or that extending the adversary tree cannot improve it, and found none across the papers above, the follow-ups of the same group, the surveys, and the works citing \cite{r1} and \cite{r2} through 2026. The nearest statement is \cite[Remark 10]{r1}: ``our choice of $w$ and $r$ are optimal in the sense that the lower bound of Lemma 8 cannot be used to prove a higher lower bound on $R$ using other values of $w$ for the formula which we obtained'', together with the remark in their Section 4 that the tightness of the analysis of this construction was verified by a mathematical program for small $t$. Both concern the analysis of the fixed construction, not the construction among extensions; the paper never discusses whether more batches, branching or adaptivity could help. For online bin \emph{stretching}, Lhomme, Catusse and Brauner \cite{r15} prove that restricting the adversary to item sizes of granularity $1/g$ yields computable lower bounds converging to the true game value, the general principle behind ``a restricted class has a computable optimum''; it has not been applied to classic bin packing, and the computer-search lower bounds of Böhm and Simon \cite{r16} and Lieskovský \cite{r17} are likewise for bin stretching only.

\emph{Upper bounds and framework floors.} The harmonic lineage runs from Lee and Lee \cite{r18} through Ramanan et al. \cite{r19}, Seiden's Super Harmonic framework and Harmonic++ ($1.58889$) \cite{r20}, Heydrich and van Stee's Extreme Harmonic \cite{r3}, and Advanced Harmonic \cite{r2}. Ramanan et al. prove that no Super Harmonic algorithm beats $19/12$; Heydrich and van Stee prove a floor for algorithms that colour a fixed fraction of each type red and decide by type. \textbf{Correction 1.} That floor is $1.5762$ in the full version of their paper (arXiv:1511.00876v7, Section 8: inputs of sizes $1/2+\varepsilon$, $1/3+\varepsilon$, $1/4+\varepsilon$, $1/7+\varepsilon$; the case of two red type-4 items per bin gives $1.5762$, the other cases $1.5788$ and $100/63$); the value $1.5766$, which is often cited, appears only in the ICALP extended abstract. \textbf{Correction 2.} Whether this floor formally covers Advanced Harmonic, which uses several container types per class with fixed proportions rather than a single red fraction, is not stated anywhere we could find; the Advanced Harmonic paper declares no floor for its own framework and, in its Appendix D, questions the $1.5815$ analysis of Son of Harmonic. It is plausible that the floor applies, since the proportions are fixed fractions, and Section 3.4 shows empirically that no further tuning gain is reachable from the tuned table, but the formal statement is open. Epstein's recent tight analysis of Harmonic$_M$ \cite{r21} concerns a single algorithm rather than a framework floor.

\emph{Has the gap moved since 2018?} Across the works citing \cite{r1} and \cite{r2} through 2026, none improves either bound; recent papers continue to cite $1.54278$ and $1.57829$ as best known.

\section{Open problems}

\begin{enumerate}
\item \textbf{Full-class optimality of the 2018 construction.} Corollary 4.4 reduces it to extensions containing window items; a proof would exhibit a single algorithm-side strategy, equivalently a dual weight function, that handles every adversary in the class, i.e. an upper-bound proof restricted to the class. All computational evidence points to the construction being optimal.
\item \textbf{Window compositions.} Two or more window batches, or a window batch together with a new statistic, are covered only computationally.
\item \textbf{A complement design for a non-opener condition.} Section 5.4 shows that the paper's complements are specific to the opener condition. Whether some other legal condition admits its own complement design that beats $r^*_\infty$ is open; the separation theorem of \cite{r1} already licenses it.
\item \textbf{Does the fixed-fraction floor cover Advanced Harmonic?} If it does, $1.5762$ is the ceiling of all proportion tuning; if not, the ceiling is unknown.
\item \textbf{A mechanism for two-medium bins.} The bottleneck of Advanced Harmonic is unaddressed by any known technique; adaptive pairing decisions are the natural candidate.
\item \textbf{Finer scales and more statistics.} Adaptivity at the $1/43$ and $1/49$ scales, and three or more statistics, are beyond the explicit-pattern machine and would need column generation; the scale-mismatch mechanism identified for split adaptive batches suggests no gain, but this is not proved.
\end{enumerate}

\section*{Acknowledgements}

The code, the experiments, the accompanying documents and the text of this paper were developed with the assistance of Claude, an AI system made by Anthropic, used through Claude Code. The assistant wrote and refactored code, ran and monitored experiments, drafted proofs and prose, and carried out the literature searches, all under the author's direction and review. Every mathematical claim was checked by the independent verifiers and the reproduction harness described in Appendix B, and the author takes full responsibility for the content. The assistant is not an author.

\appendix
\section{Exact values}

{\small\noindent \textbf{Upper bound.} Published certificate \cite{r2}: \digits{10\dgrp{}060\dgrp{}574\dgrp{}276\dgrp{}093\dgrp{}395\dgrp{}247\dslash{}6\dgrp{}374\dgrp{}352\dgrp{}691\dgrp{}333\dgrp{}693\dgrp{}440} $= 1.5782895555\ldots$. Theorem 3.1: \digits{22\dgrp{}981\dgrp{}899\dgrp{}117\dgrp{}912\dgrp{}049\dgrp{}614\dgrp{}283\dslash{}14\dgrp{}561\dgrp{}271\dgrp{}504\dgrp{}392\dgrp{}687\dgrp{}714\dgrp{}304} $= 1.5782893074\ldots$. Eleven shifts: \digits{396\dgrp{}506\dgrp{}895\dgrp{}178\dgrp{}752\dgrp{}730\dgrp{}140\dgrp{}119\dgrp{}789\dgrp{}995\dgrp{}638\dgrp{}512\dgrp{}460\dgrp{}459\dgrp{}337\dgrp{}309\dslash{}251\dgrp{}236\dgrp{}666\dgrp{}738\dgrp{}523\dgrp{}006\dgrp{}506\dgrp{}532\dgrp{}046\dgrp{}106\dgrp{}587\dgrp{}389\dgrp{}752\dgrp{}053\dgrp{}334\dgrp{}016} $= 1.5782206488\ldots$. Theorem 3.2 (sixteen shifts): \digits{4\dgrp{}352\dgrp{}511\dgrp{}675\dgrp{}791\dgrp{}652\dgrp{}445\dgrp{}491\dgrp{}291\dgrp{}971\dgrp{}161\dgrp{}065\dgrp{}859\dgrp{}763\dgrp{}057\dgrp{}971\dgrp{}646\dgrp{}201\dgrp{}076\dgrp{}221\dgrp{}224\dgrp{}375\dgrp{}333\dgrp{}444\dgrp{}301\dslash{}2\dgrp{}757\dgrp{}870\dgrp{}367\dgrp{}885\dgrp{}430\dgrp{}623\dgrp{}246\dgrp{}395\dgrp{}267\dgrp{}379\dgrp{}223\dgrp{}653\dgrp{}471\dgrp{}665\dgrp{}545\dgrp{}153\dgrp{}623\dgrp{}399\dgrp{}543\dgrp{}807\dgrp{}051\dgrp{}482\dgrp{}071\dgrp{}748} $= 1.5782147437\ldots$.\par}

{\small\noindent \textbf{Lower bound.} $r^*_\infty = (1363-\sqrt{1387369})/120 = 1.5427809064729\ldots$; $r^*(3) = 1.5427541761586\ldots$; $r^*(4) = 1.5427803609\ldots$. Machine float value of $T^*_3$: $1.5427541456$ (float bisection bias about $-3\times 10^{-8}$; every claimed value was re-checked exactly). Window: $(1/2, 1/r^*_\infty] = (0.5, 0.648193\ldots]$; $1-1/r^*_\infty = 0.351806\ldots$; $1/(2r^*_\infty) = 0.324096\ldots$.\par}

\section{Reproducibility}

{\small\noindent All code is pure-standard-library Python with exact rational arithmetic, run from the repository root as modules. Library: \code{binpacking/exact/} (infinitesimal sizes, exact simplex, maximal pattern enumeration), \code{binpacking/lower_bound/} (adversary trees and the behaviour LP in \code{adversary_tree.py}, the independent certificate checker in \code{certificate_check.py}, structural search in \code{continuation_search.py}, window certification in \code{leaf_extension.py} and \code{interval_certificate.py}, the multi-batch engine in \code{multi_batch_search.py}, adaptive conditions in \code{predicate_tree.py}), \code{binpacking/upper_bound/} (class boundaries and proportions in \code{advanced_harmonic_tables.py}, scenarios, weights, exact knapsacks, cutting-plane optimization and the feasibility screen in \code{advanced_harmonic.py}, the independent verifier in \code{certificate_check.py}, proportion states and watch lists in \code{alpha_state.py}, the marking transplant in \code{marking.py}). Runnable experiments live in \code{experiments/lower_bound/} and \code{experiments/upper_bound/}; the referee sweep is \code{python3} \code{-m} \code{experiments.upper_bound.certify_advanced_harmonic} \code{--state} \code{alpha-record.pkl} and its verification is \code{python3} \code{-m} \code{experiments.upper_bound.verify_certificates} \code{--tag} \code{e} \code{--state} \code{alpha-record.pkl}. Proportion states are \code{data/states/alpha-descent-d.pkl} (11 shifts) and \code{data/states/alpha-record.pkl} (16 shifts, Theorem 3.2); certificates are in \code{data/certificates/}. The complete experiment log is \code{docs/results-log.md}; the proof of Theorem 4.2 with full bookkeeping is \code{docs/theorem-multi-batch-extension.md}; the literature check is \code{docs/literature-survey-3-prior-art.md}. The repository is at \url{https://github}.com/ksnmartin/online-bin-packing-bounds, release tag v1.0; its README describes the layout and how to reproduce every number. The code is released under the Apache License 2.0 and this paper and the accompanying documents under CC BY 4.0.\par}

\end{document}